\documentclass[11pt]{article}
\usepackage[utf8]{inputenc} 
\usepackage[T1]{fontenc}

\usepackage{fullpage}
\usepackage[T1]{fontenc}
\usepackage{fourier}
\usepackage{authblk}

\usepackage[cmex10]{amsmath}
\usepackage{bm}
\usepackage{amssymb}

\usepackage{algorithm}
\usepackage[noend]{algpseudocode}

\usepackage{tikz}
\usepackage{makeidx}
\usepackage{graphicx}
\usepackage{multirow}
\usepackage{hhline}
\usepackage{booktabs}       
\usepackage{amsmath}

\usepackage{wrapfig}
\usepackage{graphicx}
\usepackage{caption}
\usepackage{subcaption}
\graphicspath{ {data-plots/} }
\usepackage{mathtools}

\usepackage{url}

\usepackage{url}
\usepackage{ifthen}
\usepackage{cite}
\usepackage[cmex10]{amsmath} 

\usepackage{graphicx}
\usepackage{mathtools}
\usepackage{amssymb}
\usepackage{algpseudocode}
\usepackage{relsize}

\usepackage{algorithm}
\usepackage{verbatim}
\usepackage{bm}
\usepackage[utf8]{inputenc}
\usepackage{array}
\usepackage{multirow}
\usepackage{hhline}
\usepackage{amsthm}

\graphicspath{{Dropbox/}}

\DeclareSymbolFont{matha}{OML}{txmi}{m}{it}
\DeclareMathSymbol{\varv}{\mathord}{matha}{118}

\DeclareMathOperator{\Div}{div}

\newtheorem{construction}{Construction}[section]

\newtheorem{theorem}{Theorem}[section]

\newtheorem{lemma}[theorem]{Lemma}
\newtheorem{claim}[theorem]{Claim}
\newtheorem{proposition}[theorem]{Proposition}
\newtheorem{definition}{Definition}[section]
\newtheorem{remark}{Remark}

\newcommand{\gal}{\mathrm{Gal}}

\newcommand{\supp}{{\mathrm{supp}}}
\newcommand{\wt}{{\mathrm{wt}}}

\newcommand{\me}{\mathrm{e}}

\makeatletter
\newcommand*{\rom}[1]{\expandafter\@slowromancap\romannumeral #1@}
\makeatother

\author[1]{Venkatesan Guruswami \thanks{VG's work was done while visiting NTU, Singapore.}}
\author[2]{Satyanarayana V. Lokam}
\author[3]{Sai Vikneshwar Mani Jayaraman \thanks{SVMJ's work was done during an internship at MSR, India.}}
\affil[1]{Carnegie Mellon University, USA}
\affil[2]{Microsoft Research, India}
\affil[3]{University at Buffalo, SUNY}
\title{$\epsilon$-MSR Codes: Contacting Fewer Code Blocks for Exact Repair}

\date{\vspace{-5ex}}
\begin{document}
\maketitle
\begin{abstract}
$\epsilon$-Minimum Storage Regenerating ($\epsilon$-MSR) codes form a special class of Maximum Distance Separable (MDS) codes, providing mechanisms for exact regeneration of a single code block in their codewords by downloading slighly sub-optimal amount of information from the remaining code blocks. The key advantage of these codes is a significantly lower sub-packetization that grows only logarithmically with the length of the code, while providing optimality in storage and error-correcting capacity. However, from an implementation point of view, these codes require each remaining code block to be available for the repair of any single code block. In this paper, we address this issue by constructing $\epsilon$-MSR codes that can repair a failed code block by contacting a fewer number of available code blocks. When a code block fails, our repair procedure needs to contact a few compulsory code blocks and is free to choose any subset of available code blocks for the remaining choices. Further, our construction requiresa field size linear in code length and ensures load balancing among the contacted code blocks in terms of information downloaded from them for a single repair.    
\end{abstract}
\section{Introduction}
\paragraph{} Designing error-correcting codes for distributed storage systems has evolved into an important area of research with both theoretical and practical challenges. In a typical setup, a coding scheme encodes a file of $k$ symbols over a finite field $\mathbb{F}$ (message) into a codeword comprised of $n$ symbols. These $n$ symbols are then stored on $n$ distinct storage nodes. When a single node fails, it is a repaired by \emph{regenerating} the symbol at that node from information received from data stored at the remaining $n-1$ (or fewer) \emph{intact} nodes.

\paragraph{} Low cost mechanisms to exactly regenerate the code symbols at failed nodes to replenish the redundancy are essential for sustained applicability of error correcting codes to distributed storage systems. Such low cost mechanisms also enable efficient access to the data stored at a temporarily unavailable node with the help of data stored at the remaining available nodes. An important cost metric -- introduced by Dimakis et. al (\cite{DGWR07},\cite{DGWWR10}) -- is the \emph{repair bandwidth} of a repair algorithm and is defined to be the maximum amount of data downloaded from any helper node to repair a failed node. 

\paragraph{} Vector Maximum Distance Separable (MDS) codes are often preferred for applications in distributed storage systems. In the typical scenario mentioned above, each code symbol (over $\mathbb{F}$) of a codeword is viewed as a vector of length $\ell$ over a subfield $\mathbb{B}$ of $\mathbb{F}$. The parameter $\ell$ is called the \emph{subpacketization level} or \emph{node size}. The cut-set lower bound~\cite{DGWR07} shows that the \emph{repair bandwidth} is lower bounded by	
\[ \left(\dfrac{t}{t-k+1}\right) \cdot \ell \;\;\; \text{ symbols (over} \; \mathbb{B} \text{),} \]
 assuming that the underlying code is MDS and $t$ of the $n-1$ nodes are contacted to repair a single code block. Another important consideration is \emph{load balancing} -- downloading the same amount of data from each of the $t$ helper nodes. MDS codes that achieve optimal repair bandwidth and load balancing are called \emph{Minimum Storage Regenerating} (MSR) codes in literature. They download exactly $\frac{\ell}{t-k+1}$ symbols from each of the $t$ helper nodes to repair a single failed node. A useful relaxed notion called $\epsilon$-MSR codes was defined in~\cite{RTGE17}:

\begin{definition}[$\epsilon$-MSR code] \label{Definition:emsr}
Let $\mathcal{C}$ be an $[n, k \ell, d_{\min} = n - k + 1, \ell]_{\mathbb{B}}$ vector MDS code. \footnote{See Section~\ref{subsection:vc} for the definition of vector MDS codes.} For $\epsilon > 0$, we say that the code $\mathcal{C}$ is an $\epsilon$-MSR code if there is a repair algorithm with repair bandwidth at most $(1 + \epsilon) \cdot \frac{\ell}{n - k}$ symbols (over $\mathbb{B}$).
\end{definition}

\begin{remark}
An $\epsilon$-MSR code with $\epsilon = 0$ is simply an MSR code.
\end{remark}

\paragraph{} We briefly describe the repair procedure for an MSR code here. Let's start by assuming that any single code block $\mathbf{c}_i$ fails and we need to repair it by downloading a small amount of data from the remaining code blocks $\{\mathbf{c}_{j}\}_{j \neq i}$. By setting a parameter $t: k \le t \le n - 1$, for every $i \in [n]$ and $\mathcal{R} \subseteq [n] \setminus \{i\}$ with $|\mathcal{R}| = t$, we have a collection of functions $\{h^{(i)}_{j, \mathcal{R}} : \mathbb{B}^{\ell} \mapsto \mathbb{B}^{\beta_{j, i}}\}_{j \in \mathcal{R}}$ such that $\mathbf{c}_{i}$ is a function of the symbols in the set $\{h^{(i)}_{j, \mathcal{R}} (\mathbf{c}_j)\}_{j \in \mathcal{R}}$. This implies that for every $i \in [n]$, the code block $\mathbf{c}_{i}$ can be \emph{exactly repaired} by contacting any $t$ out of $n - 1$ remaining code blocks in the codeword $\mathbf{c}$ (say indexed by the set $\mathcal{R} \subseteq [n] \setminus \{i\}$) and downloading at most $\sum_{j \in \mathcal{R}}\beta_{j, i}$ symbols of $\mathbb{B}$ from the contacted code blocks. Here, $\beta_{j, i}$ denotes the number of symbols of $\mathbb{B}$ to be downloaded from the code block $\mathbf{c}_{j}$.
\subsection{Prior Work on MSR codes and $\epsilon$-MSR codes}
\paragraph{} Constructing MSR codes that achieve optimality in all parameters in the high-rate regime is an actively pursued question, leading to an astounding line of work in the last few years. This resulted in elegant constructions tuned for specific values of $n$, $k$ and $t$ (see for example \cite{PDC11},\cite{RSK11},\cite{TWB13}). For the sake of brevity, we stick to the papers that are closely related to our paper. Ye and Barg in a recent work~\cite{YB17} constructed MSR codes that meet the optimal repair bandwidth for all values of $n$, $k$ and a fixed value $t : k \le t < n - 1$, with a sub-packetization level $\ell = (t - k + 1)^{n}$. In a follow up work, the same authors~\cite{YB16} improved $\ell$ to $(n - k)^{\frac{n}{n - k}}$ for the case when $t = n - 1$. The same result was obtained independently in~\cite{SVK16} using a coupled-layer construction. Futher, the same authors extended their results to $t < n - 1$ in~\cite{SVK17} with $\ell = (t - k + 1)^{\frac{n}{t - k + 1}}$ but their construction can only perform exact repair for a specific set of nodes. We would mention that these are the best sub-packetization levels known so far for MSR codes and a work by~\cite{GTC14} shows that an MSR code that employs linear repair schemes satisfies the following bound on its sub-packetization level:
\[k \le 2 (\log_{2}\ell) \left(\log_{\frac{n - k}{n - k - 1}} \ell + 1\right) + 1.\]

\paragraph{} To see the motivation behind $\epsilon$-MSR codes, consider a code with large sub-packetization level say $\ell \ge (n - k)^{\frac{n}{n - k}} = 2^{n/2}$. Note that this implies using storage nodes with capacity $\ell$ symbols (over $\mathbb{B}$), one can only design a storage system with only at most $2 \log_{2}{\ell}$ nodes. Thus, a larger sub-packetization level can lead to a reduced design space in terms of various system parameters. We refer the reader to~\cite{RTGE17J} for a detailed account of benefits provided by MSR codes with small sub-packetization. In order to address this issue, the authors of~\cite{GR17} introduced a variation, where they relaxed the notion that MSR code should achieve the optimal repair bandwidth. Interestingly, they observed that losing an additional factor of $\epsilon$ in the repair-bandwidth leads to significant gains in terms of subpacketization i.e. $\ell = (n - k)^{\frac{1}{\epsilon}}$ but their construction suffered from a large alphabet size and was not explicit. They addressed this issue in a followup work~\cite{RTGE17}, presenting an explicit construction and introduced the notion of $\epsilon$-MSR codes. Further, the sub-packetization of their code grows only logarithmically with the length. Moreover, they introduced a general framework for constructing $\epsilon$-MSR codes, which involves combining an existing MSR code with another linear code of large distance. We work with this framework throughout the paper.  We would like to mention here that~\cite{RTGE17} and~\cite{GR17} only consider the $t = n - 1$ case.

\subsection{Our Contributions}
\paragraph{} In this paper, we construct $\epsilon$-MSR codes that contact a fewer number of helper nodes, answering an important open question posed in~\cite{RTGE17},~\cite{RTGE17J} and~\cite{GR17}. To this end, we first introduce the notion of $(\mathcal{T}, \mathcal{T}')$-repair property for $\epsilon$-MSR codes.

\begin{definition} [$(\mathcal{T}, \mathcal{T}')$-repair property] \label{Definition:torp}
For a fixed value of $k \le \mathcal{T} < n - 1$, an $\epsilon$-MSR code $\mathcal{C} = [n, k\ell, d_{\min} = n - k + 1,\ell = s^{n}]$ with $s = \mathcal{T} - k + 1$ has the {\em $(\mathcal{T}, \mathcal{T'})$-repair property} if and only if the exact repair of any node $i \in [n]$ can be accomplished by receiving at most $(1 + \epsilon) \cdot \frac{\ell}{s}$ symbols (over $\mathbb{B}$) from two types of helper nodes: $\mathcal{T}'$ \emph{compulsory} nodes that always need to be chosen and the remaining $\mathcal{T} - \mathcal{T}'$ nodes can be chosen arbitrarily.   \end{definition}

\paragraph{} Our construction satisfies this property and is an instantiation of the $\epsilon$-MSR code framework introduced in~\cite{RTGE17}. At a high level, their framework combines an \emph{inner} MSR code with an \emph{outer} linear code of large distance in order to obtain an $\epsilon$-MSR code. We choose the inner MSR code to be Construction \rom{4} in~\cite{YB17} that satisfies the $t$-optimal repair property. Further, we choose the outer code to be an Algebraic Geometry (AG) code with specific parameters that minimizes the number of compulsory nodes (see Theorem~\ref{theorem:ag1}). Finally, we obtain an $\epsilon$-MSR code with $(\mathcal{T}, \mathcal{T}')$-repair property whose sub-packetization grows only {\em logarithmically} with the code length and has a constant number of parity symbols (see Theorem~\ref{Lemma:inffamily}). This amounts to a doubly-exponential saving in terms of the sub-packetization level compared to the existing MSR codes with $t$-optimal repair property (\cite{SVK17},\cite{YB17}). Similar to these two codes, our constructions are over an alphabet linear in the length of the code, can be generalized to an infinite code family and ensure load balancing among the contacted nodes.
\subsection{Organization}
\paragraph{} In Section~\rom{2}, we present the neccessary background on \emph{MDS vector codes} along with a parity-check view of them. Further, we recap the $\epsilon$-MSR code framework defined in~\cite{RTGE17}, which we use for our construction. We present the main result of this paper in Section~\rom{3} giving an explicit family of $\epsilon$-MSR codes that contacts a fewer number of helper nodes. We conclude the paper in Section~\rom{4} with directions for future work.
\section{Preliminaries and Notation}
\subsection{Vector MDS codes} \label{subsection:vc}
\paragraph{} We consider a set $\mathcal{C} \subseteq \mathbb{F}^{n}$ over a finite field $\mathbb{F}$. We say that $\mathcal{C}$ forms an $(n, M, d_{min})_{\mathbb{F}}$ code if we have $|\mathcal{C}| = M$ and $d_{min} = \min_{\mathbf{c} \neq \mathbf{c'} \in \mathcal{C}} d_{H}(\mathbf{c}, \mathbf{c'})$, where $d_{H}(., .)$ denotes the Hamming distance.

\paragraph{} In this paper, we consider $\mathbb{F}$ to be a degree-$\ell$ extension of a subfield $\mathbb{B}$. Hence, each element of $\mathbb{F}$ can be represented as an $\ell$-length vector over $\mathbb{B}$. It follows that a codeword $\mathbf{c} = (c_{1}, \dots, c_{n}) \in \mathcal{C} \subseteq \mathbb{F}^{n}$ can be represented as an $n\ell$-length vector $\mathbf{c} = (\mathbf{c}_{1}, \dots, \mathbf{c}_{n}) \in \mathbf{B}^{n\ell}$. Here, for $i \in [n]$, the \emph{code block} $\mathbf{c}_{i} = (c_{i, 1}, \dots, c_{i, \ell}) \in \mathbb{B}^{\ell}$ denotes the $\ell$-length vector corresponding to the code symbol $c_{i} \in \mathbb{F}$. In this setting, we say $\mathcal{C}$ is a \emph{vector linear code}; it forms a $\log_{|\mathbb{B}|}M$-dimensional subspace of $\mathbb{B}^{n\ell}$ (over $\mathbb{B}$). We say that an $[n, \log_{|\mathbb{B}|}{M}, d_{\min}, \ell]_{\mathbb{B}}$ \emph{vector linear code} is a \emph{vector MDS code} if $\ell$ divides $\log_{|\mathbb{B}|}M$ and $d_{\min} = n - \frac{\log_{|\mathbb{B}|}M}{\ell} + 1$.
\begin{remark}
Note that even though we view the codewords of an array code as $n \ell$-length vectors over $\mathbb{B}$, the minimum distance is calculated by viewing each code block as a symbol over $\mathbb{F}$. Therefore, the minimum distance of the array code belongs to the set of integers $[n] = \{1, \dots, n\}$.
\end{remark}
\paragraph{} A $[n, \log_{|\mathbb{B}|}M, d_{\min}, \ell]_{\mathbb{B}}$ \emph{vector linear code} can be defined by an $(n \ell - \log_{|\mathbb{B}|}M) \times n \ell$ full rank matrix $\mathbf{H}$ over $\mathbb{B}$ as follows:
\begin{equation}
\mathcal{C} = \{\mathbf{c} = (\mathbf{c}_{1}, \dots, \mathbf{c}_{n}) : \mathbf{H} \cdot \mathbf{c} = 0\} \subseteq \mathbb{B}^{n \ell}.
\end{equation}
\paragraph{} The matrix $\mathbf{H}$ is called the \emph{parity check matrix} of the code $\mathcal{C}$. Assuming that $k$ is an integer such that $\log_{|\mathbb{B}|} M = k \ell$, we can view $\mathbf{H}$ as a block matrix
\begin{align*}
\mathbf{H} = \begin{pmatrix}
\mathbf{H}_{1} & \mathbf{H}_{2} & \dots & \mathbf{H}_{n}
\end{pmatrix}
\in \mathbb{B}^{(n - k)\ell \times n \ell}.
\end{align*}
\paragraph{} For $i \in [n]$, we refer to the $(n - k) \ell \times \ell$ sub-matrix $\mathbf{H}_{i}$ as the \emph{thick column} associated with the $i$-th code block in the codewords of $\mathcal{C}$. For a set $S = \{i_1, i_2, \dots, i_{|S|} \subseteq [n]\}$, we define the $(n - k)\ell \times |S|\ell$ matrix $H_{S}$ as follows,
\begin{align*}
\mathbf{H}_{S} = \begin{pmatrix}
\mathbf{H}_{i_1} & \mathbf{H}_{i_2} & \dots & \mathbf{H}_{i_{|S|}} 
\end{pmatrix}
\in \mathbb{B}^{(n - k) \ell \times |S| \ell}.
\end{align*}
\paragraph{} Note that the matrix $\mathbf{H}_{S}$ comprises the thick columns with indices in the set $S$. The parity-check matrix $\mathbf{H}$ defines a \emph{vector MDS code} if for every $S \subseteq [n]$ with $|S| = n - k$, the $(n - k)\ell \times (n - k)\ell$ sub-matrix $\mathbf{H}_{S}$ is full rank. 
\subsection{$\epsilon$-MSR code framework}
\paragraph{} In this section, we describe the framework for constructing $\epsilon$-MSR codes defined in~\cite{RTGE17}.
\begin{construction} \label{Construction:main}
	We are given two codes:
	\begin{itemize}
		\item{An $(n, r = n - k, t, \ell)_{\mathbb{B}}$ MSR code $\mathcal{C}^{I}$ with length $n$, number of message symbols $k$, number of parity symbols $r$, number of helper nodes $t$ to be contacted for exact repair of any single node and sub-packetization $\ell$ is defined by the parity-check matrix
			\begin{equation} \label{eq:main_pc1}
			\mathbf{H} =	
			\begin{pmatrix}
			H_{1, 1} & H_{1, 2} & \dots & H_{1, n} \\
			H_{2, 1} & H_{2, 2} & \dots & H_{2, n}  \\
			\vdots & \vdots & \vdots & \vdots \\
			H_{r, 1}& H_{r, 2} & \dots & H_{r, n} \end{pmatrix}
			\in \mathbb{B}^{r\ell \times n\ell}.
			\end{equation}
		}
		\item{A $(N, K, M, D = \delta N)_{q}$ linear code $\mathcal{C}^{II}$ with length $N$, size $M$ and distance $\delta N$ over the alphabet $q \le n$.}
	\end{itemize}
	The $(\mathcal{N} = M, (M - r)N\ell, \mathcal{T}, \mathcal{L} = N\ell)_{\mathbb{B}}$ $\epsilon$-MSR code $\mathcal{C} = \mathcal{C}^{II} \circ \mathcal{C}^{I}$ is constructed by designing its $rN\ell \times MN\ell$ parity check matrix $\mathcal{H}$. Notice that every codeword of $\mathcal{C}$ comprises of $M  = |\mathcal{C}^{II}|$ code blocks with each of these blocks containing $N\ell$ symbols (over $\mathbb{B}$). In particular, this implies that the $M$ code blocks in a codeword of $\mathcal{C}$ are indexed by $M$ distinct $N$-length codewords in $\mathcal{C}^{II}$. Let $\mathbf{a}_{i} = (a_{i, 1}, \dots, a_{i, N}) \in q^N$ be a codeword of $\mathcal{C}^{II}$. Then, the $N\ell$ columns of the parity check matrix $\mathcal{H}$ that correspond to the code block of a codeword of $\mathcal{C}$ indexed by $\mathbf{a}_{i} \in \mathcal{C}^{II}$ are defined as follows:
	\begin{equation} \label{eq:main_pc2}
	\mathcal{H}_{i} =
	\begin{bmatrix}
	Diag(H_{1, a_{i, 1}}, \dots, H_{1, a_{i, N}}) \\
	\sigma_{i} \cdot Diag(H_{2, a_{i, 1}}, \dots, H_{2, a_{i, N}})  \\
	\vdots \\
	\sigma_{i}^{r - 1} \cdot Diag(H_{r, a_{i, 1}}, \dots, H_{r, a_{i, N}}) \\
	\end{bmatrix},
	\end{equation}
	where $\{\sigma_{i}\}_{i \in [M]}$ are non-zero distinct elements suitably chosen from $\mathbb{B}$. It follows that all the blocks \\ $\{H_{j, a_{i, k}}\}_{j \in [r], i \in [M], k \in [N]}$ in the above equation are well-defined since $q \le n$.
\end{construction}	
\section{$\epsilon$-MSR code construction with $(\mathcal{T}, \mathcal{T}')$-repair property}
\paragraph{} In this section, we present an explicit construction of an $\epsilon$-MSR code having the $(\mathcal{T}, \mathcal{T}')$-repair property. We start by describing our construction, which leverages the $\epsilon$-MSR code framework defined earlier. Then, we explain our repair procedure in Claim~\ref{Claim:single-node-repair}, which we use to argue that our code is an $\epsilon$-MSR code and satisifes the $(\mathcal{T}, \mathcal{T}')$-repair property. We then record our main result in Theorem~\ref{Lemma:inffamily}, presenting an infinite family of such codes. Finally, we establish that $\mathcal{C}$ is a MDS code in Appendix~\ref{LemmaRef:d11}.
\begin{construction} \label{d-opt-ag}
The $(\mathcal{N} = M, (M - r)N \ell, \mathcal{T} = M - n + t, \mathcal{L} = N \ell)_{\mathbb{B}}$ $\epsilon$-MSR code with the $(\mathcal{T}, \mathcal{T}')$-repair property is constructed using Construction~\ref{Construction:main} by choosing:
\begin{itemize}
\item {The code $\mathcal{C}^{I}$ to be the $(n, r = n - k, k \le t < n - 1, \ell = s^{n})_{\mathbb{B}}$-MSR code having $t$-optimal repair property (where $s = t - k + 1$) for a fixed value of $t$, defined by
\begin{equation} \label{eq:ybs-pc2}
H_{j, i} = H_{i}^{j - 1}, j \in [r], i \in [n].
\end{equation}
Here, $H_{i} = \sum_{b = 0}^{\ell - 1}\lambda_{i, b_i} \cdot e_b \cdot e_{b}^{T}$ for every $i \in [n]$. Note that $\mathcal{C}^{I}$ can be constructed over a finite field of size $|F| \ge sn$ and $\{e_b : b \in [0, \ell - 1]\}$ is the standard basis of $F^{\ell}$ over $F$. Further, $b_j$ is the $j$-th digit from the right in the representation of $b$ in the $s$-ary form, $b = (b_n, b_{n  -1}, \dots, b_1)$. (Details in Appendix~\ref{Ye-Barg-d-optimal}.)}
\item {The code $\mathcal{C}^{II}$ to be a $(N, K, M = q^{K}, D = \delta N)_{q}$-linear code with constant alphabet size $r < q \le n$, constant relative distance $\delta$ and the number of codewords with Hamming Weight $N$ is at least $q^{\frac{u - 1}{u} K}$ for any $u > 3$. We provide an explicit family of such codes in Theorem~\ref{theorem:ag1}.}
\end{itemize}
\end{construction}
\paragraph{} From the above Construction, the $M$ code blocks in any codeword of $\mathcal{C}$ are indexed by $M$ distinct codewords in the code $\mathcal{C}^{II}$. Let's assume that $\mathcal{C} = (c^{1}, c^{2}, \dots, c^{M})$, where $c^{i} = (c^{i}_{1}, c^{i}_{2}, \dots, c^{i}_{N})$ for every $i \in [M]$ and $c^{i}_{j} = (c^{i}_{j, 1}, \dots, c^{i}_{j, \ell})$ for every $j \in [N]$.  Without loss of generality, we assume that the first code block $c^{1} \in \mathcal{C}$ fails and it is indexed by the codeword $\mathbf{a}_{1} = (a_{1, 1}, \dots, a_{1, N}) \in \mathcal{C}^{II}$. 

\paragraph{} For any codeword $\mathbf{a}_{i} = (a_{i, 1}, a_{i, 2}, \dots, a_{i, N}) \in \mathcal{C}^{II}, i \in [2, M]$, the corresponding parity check column $H_{i}$ of $\mathcal{C}$ indexed by $\mathbf{a}_{i}$ is given by~\eqref{eq:main_pc2}. Since all the parity check columns are block diagonal matrices, we first repair $c^{1}_{1} \in c^{1}$  and by symmetry, we can repair all $c^{1}_{j} : j \in [2, N]$ in a similar fashion. We start with the following claim on the repair bandwidth of $c^{1}_{1}$: 
\begin{claim} \label{Claim:single-node-repair}
The code block $c^{1}_{1}$ can be repaired by downloading
\[\left (\frac{M}{q} - 1 \right )\ell + \left (\mathcal{T} - \left (\frac{M}{q} - 1 \right ) \right )\dfrac{\ell}{s}\]
symbols (over $\mathbb{B}$) from $\mathcal{T}$ code blocks $c^{i}_{1}, i \in [2, M]$. In particular, for every $i \in [2, M]$, the $\frac{M}{q} - 1$ code blocks $c^{i}_{1}$, whose corresponding index $a_{i, 1}$ equals $a_{1, 1}$ need to be contacted compulsorily. The remaining $\mathcal{T} - \left(\frac{M}{q} - 1\right)$ code blocks can be chosen arbitrarily.
\end{claim}
\begin{proof}
To this end, consider the $M\ell$ columns of $\mathcal{H}$ that correspond to $a_{i, 1} : i \in [M]$:
\begin{equation} \label{eq:ourd-pc4}
\begin{bmatrix}
H_{1, a_{1, 1}} & H_{1, a_{2, 1}} & \dots & H_{1, a_{M, 1}}\\
\sigma_{1} H_{2, a_{1, 1}} & \sigma_{2} H_{2, a_{2, 1}} & \dots & \sigma_{M} H_{2, a_{M, 1}}\\
\vdots & \vdots & \vdots & \vdots \\
\sigma_{1}^{r - 1} H_{r, a_{1, 1}} & \sigma_{2}^{r - 1} H_{r, a_{2, 1}} & \dots & \sigma_{M}^{r - 1} H_{r, a_{M, 1}}
\end{bmatrix}.
\end{equation}
Using our definition of $\mathcal{C}$, we can write the $r\ell$ parity check equations corresponding to $a_{i, 1}: i \in [M]$ using~\eqref{eq:ourd-pc4} as follows:
\begin{equation} \label{eq:ourd-pc5}
\sum_{i = 1}^{i = M} \sigma_{i}^{j - 1} \cdot H_{j, a_{i, 1}} \cdot c^{i}_{1} = 0 \text{ for every } j \in [r].
\end{equation}
We now build two sets $Q$, $V$ and a tuple $\Gamma$, where $Q = \{i : a_{i, 1} = a_{1, 1}, i \in [2, M]\}$, $V = \{i : a_{i, 1} \neq a_{1, 1}, i \in [2, M]\}$ and $\Gamma = (a_{i, 1} : a_{i, 1} \neq a_{1, 1}, i \in [2, M])$. From our construction of $\mathcal{C}^{II}$, it follows that $|Q| = \frac{M}{q} - 1$, $|V| = M - \frac{M}{q}$ and for each $\gamma \in \Gamma$, we have $1 \le \gamma \le n$. (The value of $|Q|$ follows from the observation that there exists no $j \in [N]$ such that $a_{i, j} = 0$ for every $i \in [M]$, which in turn follows from the construction of $\mathcal{C}^{II}$.) 
Let $Q = \{q_i\}_{i \in [|Q|]}$, $V = \{v_i\}_{i \in [|V|]}$ and $\Gamma = (\gamma_i)_{i \in [|\Gamma|]}$. Based on this classification, we can write~\eqref{eq:ourd-pc5} as
\begin{align*}
\sigma_{1}^{j - 1} \cdot H_{j, a_{1, 1}} \cdot c^{1}_{1} + \sum_{q_i \in Q} \sigma_{q_i}^{j - 1} \cdot H_{j, a_{q_i, 1}} \cdot c^{q_i}_{1} + \sum_{v_i \in V} \sigma_{v_i}^{j - 1} \cdot H_{j, a_{v_i, 1}} \cdot c^{v_i}_{1} = 0 \text{ for every } j \in [r].
\end{align*}
We first recover the $s$ symbols $(c^{1}_{1, b(a_{1, 1}, k)})_{k \in [0, s - 1]}$ for a fixed value of $b$ in the range $[0, \ell - 1]$. Note that $b$ is a $s$-ary representation of length $n$ and $b(i, u) = (b_{n}, b_{n - 1}, \dots, b_{i + 1}, u, b_{i - 1}, \dots, b_{1})$. Consider the $rs$ parity check equations corresponding to these $s$ symbols, given by
\begin{align*} 
\sigma_{1}^{j - 1} \lambda_{a_{1, 1}, k}^{j - 1} c^{1}_{1, b(a_{1, 1}, k)} + \sum_{q_{i} \in Q}  \sigma_{q_i}^{j - 1} \lambda_{a_{1, 1}, k}^{j - 1} c^{q_i}_{1, b(a_{1, 1}, k)}  + \sum_{i = 1}^{|V|} \sigma_{v_i}^{j - 1}  \lambda_{a_{v_{i, 1}}, b_{\gamma_{i}}}^{j - 1} c^{v_i}_{1, b(a_{1, 1}, k)} = 0 \textbf{\space} \forall  j \in [r], k \in [0, s - 1], \tag{5.5} \label{eq:ourd-pc7}
\end{align*} 
following from the definition of $\{H_{i, j}\}_{i \in [n], j \in [r]}$ in Construction~\ref{d-opt-ag}. We sum the equation~\eqref{eq:ourd-pc7} over all values of $k \in [0, s - 1]$ to get
\begin{align*} 
\sum_{k= 0}^{s - 1} \sigma_{1}^{j - 1} \lambda_{a_{1, 1}, k}^{j - 1} c^{1}_{1, b(a_{1, 1}, k)} + \sum_{k = 0}^{s - 1} \sum_{q_{i} \in Q}  \sigma_{q_i}^{j - 1} \lambda_{a_{1, 1}, k}^{j - 1} c^{q_i}_{1, b(a_{1, 1}, k)} + \sum_{k = 0}^{s - 1} \sum_{i = 1}^{|V|} \sigma_{v_i}^{j - 1} \lambda_{a_{v_{i, 1}}, b_{\gamma_{i}}}^{j - 1} c^{v_i}_{1, b(a_{1, 1}, k)} = 0 \text{\space} \forall  j \in [r].
\end{align*}
Propagating the sum inside, we get
\begin{align*} 
\sigma_{1}^{j - 1} \sum_{k= 0}^{s - 1}  \lambda_{a_{1, 1}, k}^{j - 1} c^{1}_{1, b(a_{1, 1}, k)} + \sum_{q_{i} \in Q}  \sigma_{q_i}^{j - 1} \sum_{k= 0}^{s - 1} \lambda_{a_{1, 1}, k}^{j - 1} c^{q_i}_{1, b(a_{1, 1}, k)} + \sum_{i = 1}^{|V|} \sigma_{v_i}^{j - 1} \lambda_{a_{v_{i, 1}}, b_{\gamma_{i}}}^{j - 1} \sum_{k= 0}^{s - 1}  c^{v_i}_{1, b(a_{1, 1}, k)} = 0 \text{\space} \forall j \in [r].
\end{align*}
For every $v_i \in V$, let $\mu_{v_i, 1, 1}^{(b)} := \sum_{k = 0}^{s - 1} c^{v_i}_{1, b(a_{1, 1}, k)}$. We can substitute this in the above equation to get
\begin{align*} 
\sigma_{1}^{j - 1} \sum_{k= 0}^{s - 1}  \lambda_{a_{1, 1}, k}^{j - 1} c^{1}_{1, b(a_{1, 1}, k)} + \sum_{q_{i} \in Q}  \sigma_{q_i}^{j - 1} \sum_{k= 0}^{s - 1}  \lambda_{a_{1, 1}, k}^{j - 1} c^{q_i}_{1, b(a_{1, 1}, k)} + \sum_{i = 1}^{|V|} \sigma_{v_i}^{j - 1} \lambda_{a_{v_{i, 1}}, b_{\gamma_{i}}}^{j - 1} \mu_{v_i, 1, 1}^{(b)} = 0 \text { for every } j \in [r].
\end{align*}
We can write the above $r$ equations in the following matrix form: 
\begin{equation} \label{eq:ourd-pc12}
L_1 + L_2 + L_3 = 0,
\end{equation}
where $L_{1} = E_{L_1} F_{L_1}$, $L_2  = \sum_{q_i \in Q} E_{q_i} F_{q_i}$ and $L_3 = E_{V} F_{V}$.
In particular, $E_{L_1} = (\sigma_{1}^{j - 1} \lambda_{a_{1, 1}, k}^{j -1})_{j \in [r], k \in [0, s - 1]}$, $F_{L_1} = (c^{1}_{1, b(a_{1,1}, k)})_{k \in [0, s - 1]}$, $E_{q_i} =   (\sigma_{q_i}^{j - 1} \lambda_{a_{1, 1}, k}^{j -1})_{j \in [r], k \in [0, s - 1]}$, $F_{q_i} = (c^{q_i}_{1, b(a_{1, 1}, k)})_{k \in [0, s - 1]}$, $E_{V} = (\sigma_{v_k}^{j - 1} \lambda_{\gamma_k, b_{\gamma_{k}}}^{j - 1})_{j \in [r], k \in [1, |V|]}$ and $F_{V} = (\mu^{(b)}_{v_{k}, 1 , 1})_{k \in [1, |V|]}$.

We use a polynomial interpolation argument on~\eqref{eq:ourd-pc12} inspired by the proof of Theorem $7$ in~\cite{YB17}. We provide an overview of our interpolation argument here (the details are in Appendix~\ref{LemmaRef:PI1}). Recall that our goal is to compute the number of symbols we need to download from $\{F_{q_i}\}_{q_i \in Q}$ and $F_{V}$ in order to recover $F_{L_1}$. We start by showing that $F_{V}$ can be recovered from any of its $|V| - (r - s)$ symbols (note that $|V| - (r - s) > 0$ since $q > r$). Then, we use $\{F_{q_i}\}_{q_i \in Q}$ and $F_{V}$ in order to recover $F_{L_1}$. From our argument, it follows that by downloading any $|V| - (r - s)$ symbols from $F_{V} = (\mu_{v_{k}, 1, 1}^{(b)})_{k \in [1, |V|]}$, we can recover the remaining $(r - s)$ symbols and thus the whole vector $F_{V}$. Moreover, we download all the $s$ symbols in $F_{q_i}$ for all $q_i \in Q$. Thus, to determine $F_{L_1}$ (i.e. recover $s$ symbols), we need to download a total of
\[|Q|s + (|V| - (r - s))\]
symbols (over $\mathbb{B}$). In particular, this implies that to recover all the $\ell$ symbols of $c^{1}_{1}$ in a similar fashion, we need to download a total of
\[|Q|\ell + (|V| - (r - s))\dfrac{\ell}{s}\]
symbols (over $\mathbb{B}$).  Notice that we need to compulsorily contact all the $\left(\dfrac{M}{q} -1\right)$ code blocks characterized by $Q$ and can choose any arbitrary subset of size $|V| - (r - s)$ for the remaining choices. Substituting the values of $|Q|$ and $|V|$ from our definition along with $r = n - k$ and $s = t - k + 1$, we can repair $c^{1}_{1}$ by downloading
\begin{align*}
\left (\dfrac{M}{q} - 1 \right ) \ell + \left (M - n + t - \left (\dfrac{M}{q} - 1 \right ) \right )\dfrac{\ell}{s}
\end{align*}
symbols (over $\mathbb{B}$). It follows that we contact $\mathcal{T} = M - n + t$ blocks in total, completing the proof of Claim~\ref{Claim:single-node-repair}.
\end{proof}
Using the above repair procedure, we now argue that $\mathcal{C}$ has near-optimal repair bandwidth.
\begin{lemma} \label{Lemma:dmain1}
Suppose $\delta \ge 1 - \frac{\epsilon}{r - 1}$. Then, $\mathcal{C}$ is an $\epsilon$-MSR code with $(\mathcal{T} = M - n + t, \mathcal{T}')$-repair property, where $\mathcal{T}' = M - |\{\mathbf{a} \in \mathcal{C}^{II} | \wt(\mathbf{a}) = N\}|$. 
\end{lemma}
\begin{proof}
We assume without loss of generality that $c^{1}$ is the failed code block. We start by making the following observation based on our repair mechanism. In particular, our repair mechanism needs to contact all code blocks $c^{j}: j \in [2, M]$ indexed by a codeword $\mathbf{a}_{j} = (a_{j, 1}, \dots, a_{j, N}) \in \mathcal{C}^{II}$ such that there exists at least one $k \in [N]$ with $a_{1, k} = a_{j, k}$. By definition, $\mathcal{T}'$ is the number of such \emph{compulsory} blocks.
 
Assuming that we contact all the $\mathcal{T}'$ blocks and pick the rest arbitrarily, the total repair bandwidth for repairing $c^{1}$ is
\[\left (\dfrac{M}{q} - 1 \right )N\ell + \left (M - n + t - \left (\dfrac{M}{q} - 1 \right ) \right )\dfrac{N\ell}{s} \]
symbols (over $\mathbb{B}$). To prove that $\mathcal{C}$ is an $\epsilon$-MSR code with the $(\mathcal{T}, \mathcal{T}')$-repair property, we still need an upper bound on the download from each contacted code block. For this purpose, we can rewrite the above expression assuming $P$ denotes the set of contacted code blocks (i.e. $|P| = M - n + t$):
\begin{equation}
\sum_{i \in P}|\{j \in [1, N]: a_{i, j} = a_{1, j}\}| \ell + |\{j \in [1, N]: a_{i, j} \neq a_{1, j}\}| \dfrac{\ell}{s}
\end{equation}
symbols (over $\mathbb{B}$). Notice that we have now bounded the download from each code block. Now, we can do a similar analysis as in the proof of Theorem~\rom{3}.$1$ in~\cite{RTGE17} to show that the download from each contacted code block is upper bounded by $(1 + \epsilon) \frac{N\ell}{s}$ symbols (over $\mathbb{B}$) as long as $\delta \ge 1 - \frac{\epsilon}{r - 1}$. Recall that we contact only $\mathcal{T}$ code blocks out of which $\mathcal{T}'$ compulsory and download an equal amount of data from each contacted block, completing the proof.
\end{proof}
\paragraph{} Naturally, we would like $\mathcal{T}'$ to as small as possible, which translates to $\mathcal{C}^{II}$ having many codewords with Hamming weight $N$. We now construct $\mathcal{C}^{II}$ with at least $q^{\frac{u - 1}{u} K}$ ($u > 3$) such codewords, using ideas from Proposition $3.1$ in~\cite{JX12}. In particular, with growing $u$, the number of codewords with Hamming weight $N$ becomes arbitrarily close to $q^{K}$. We provide a brief introduction to AG-codes in Appendix~\ref{agcodes}.
\begin{theorem} \label{theorem:ag1}
Let $u > 3$ be a positive integer and $q$ be a square prime power greater than $2 (u + 1)^{2} \frac{r^2}{\epsilon^{2}}$ for some $\epsilon > 0$. Then, we can construct an explicit family of AG codes $\mathcal{C}^{II} = (N, K = u g, M = q^{K}, D = \delta N)_{q}$, where $\frac{g}{N} \approxeq \frac{1}{\sqrt{q} - 1}$ and $\delta \ge 1 - \frac{\epsilon}{r - 1}$, containing at least $q^{\frac{u - 1}{u} K}$ codewords with Hamming weight $N$.
\end{theorem}
\begin{proof}
	Let $\mathcal{C}^{II}$ be a $C_{\mathcal{L}}(R, G, \mathbf{v})$ AG code with $\ell{(G)} = m + 1 - g$, where $m = \deg{(G)}$. By Proposition~\ref{prop:r-r}, it follows that $\mathcal{C}^{II}$ is a $(N, K, M = q^{K}, D)$ code with $2g - 1 < m < N$ and $D \ge N - m$. Note that we need to have $K = ug$ and $D = \delta N$, where $\delta \ge 1 -\frac{\epsilon}{r - 1}$ and we assume that $\frac{g}{N} \approxeq \frac{1}{\sqrt{q} - 1}$ (follows from the lower bound of~\cite{VD83}). 
	
	For a codeword $(P_1(w), \dots, P_{n}(w))$ in $C_{\mathcal{L}}(R, G, \mathbf{1})$ with $w \in \mathcal{L}(G - \sum_{i = 1}^{N} P_i)$, the $j$-th co-ordinate $P_j(w)$ is zero if and only if $w \in \mathcal{L}(G - \sum_{i = 1}^{N}P_i + P_j)$. Thus, the number of codewords with Hamming weight $N$ in $C_{\mathcal{L}}(R, G, \mathbf{1})$ is the size of 
	\begin{equation} \label{eq:agcode1}
	\mathcal{L}(G - \sum_{i = 1}^{N} P_i) \setminus \bigcup_{j = 1}^{N} \mathcal{L}(G - \sum_{i = 1}^{N}P_i + P_j).
	\end{equation}
	We denote this set by $A$ and denote $\mathcal{L} (G - \sum_{i = 1}^{N}P_i + P_j)$ by $A_{j}$ for every $j \in [N]$. Note that the distance $D$ of $\mathcal{C}^{II}$ is at least $N - m \ge N - (K + g - 1) = N - K - g + 1$, where the second inequality follows from $m \le K + g - 1$. Using the inclusion-exclusion principle, we can estimate~\eqref{eq:agcode1} as follows:
	\begin{equation} \label{eq:agcode2}
	|A| = |\mathcal{L}(G - \sum_{i = 1}^{N} P_i)| - \sum_{j = 1}^{N} |A_j|  + \dots + (-1)^{K - g + 1} \sum_{j_1, \dots, j_{K - g + 1} \in [N], j_1 \neq \dots \neq j_{K - g + 1}} |\cap_{\alpha = 1}^{K - g + 1} A_{j_\alpha}|. 
	\end{equation}
	Using the minimum distance property, we have 
	\[ N -D \le  N - (N - (K - g + 1)) = K - g + 1.\]
	In particular, this implies $|\cap_{\alpha = 1}^{p} A_{j_\alpha}| = 0$ for any $p > K - g + 1$ and $j_{1}, \dots, j_{p} \in [N], j_{1} \neq \dots \neq j_{p}$. We still need a bound on $|\cap_{\alpha = 1}^{p} A_{j_{\alpha}}|$ for every $p \in [K - g + 1]$, which we obtain using Proposition~\ref{prop:r-r}. Consider $C^{II}$ to be shortened by $p$ bits leading to a $(N - p, K - p, D)$ code. In order for~\eqref{eq:agdimension} to be satisified and $\mathcal{C}^{II}$ to be an AG-code, we require that $1 \le p \le K - g$. Thus, we have $|\cap_{\alpha = 1}^{p} A_{j_{\alpha}}| = q^{K - p}$ for every $p \in [K - g + 1]$, which follows from the size of the shortened code. We substitute this in~\eqref{eq:agcode2} and assume that $K - g + 1$ is even so that we can ignore the last term  to get
	\begin{align*}
	|A| \ge q^{K} - \binom{N}{1}q^{K - 1} + \binom{N}{2}q^{K - 2} - \dots + (-1)^{K - g} \binom{N}{K - g} = q^{K} \left (1 - \dfrac{1}{q} \right )^{N} - \sum_{p = K - g + 1}^{N} (-1)^{p} \binom{N}{p} q^{K - p}.
	\end{align*}
	Consider the alternating sum $\sum_{j = K - g + 1}^{N} (-1)^{j} \binom{N}{j} q^{K - j}$. To approximate this sum, we consider the largest ratio between any two successive terms $y$ and $y + 1$. More formally, we have
	\begin{align*}
	\dfrac{s_{y + 1}}{s_{y}} & = \dfrac{\binom{N}{j + 1} q^{K - (j + 1)}}{\binom{N}{j} q^{K - j}} = \dfrac{N - j}{(j + 1) q} \le \dfrac{N - (K - g + 1)}{(K - g + 2) q} < 1.
	\end{align*}
	Notice that the final inequality is true only if 
	\begin{equation} \label{eq:qinequality1}
	q > \dfrac{N - (K - g + 1)}{K - g + 2}.
	\end{equation}
	Recall our earlier assumption that $K = ug$, which we apply in the above equation to get
	\[ q > \dfrac{N - (u - 1)g - 1}{(u - 1)g + 2}.\]
	In particular, as long as $q > \frac{N}{(u - 1)g}$, the above inequality is always true since $\frac{g}{N} \approxeq \frac{1}{\sqrt{q} - 1}$. Notice that this implies $\frac{s_{y + 1}}{s_{y}} < 1$ for every $K - g < j \le N$, which in turn leads to
	\begin{equation}
	\sum_{j = K - g + 1}^{N} (-1)^{j} \binom{N}{j} q^{K - j} \le \binom{N}{K - g + 1} q^{g} \le \left(\dfrac{N \me}{K - g + 1} \right)^{K - g + 1} q^{g}.
	\end{equation}
	Observe that we can now lower bound $|A|$ as follows:
	\begin{align*}
	|A| & \ge  q^{K} \left (1 - \dfrac{1}{q} \right )^{N} -   \sum_{j = K - g + 1}^{N} (-1)^{j} \binom{N}{j} q^{K - j}\\
	& \ge  q^{K} \left (1 - \dfrac{1}{q} \right )^{N} -  \left(\dfrac{N \me}{K - g + 1} \right)^{K - g + 1} q^{g}.
	\end{align*}
	We can substitute $K = ug$ in the above inequality to get
	\begin{align*}
	|A|	& \ge  q^{ug} \left (1 - \dfrac{1}{q} \right )^{N} -  \left(\dfrac{N \me}{(u - 1)g + 1} \right)^{(u - 1)g + 1} q^{g} \\
	&  \ge q^{ug} \left (1 - \dfrac{1}{q} \right )^{N} -  \left(\dfrac{N \me}{(u - 1)g} \right)^{(u - 1)g + 1} q^{g}.
	\end{align*}
	Again, using $\frac{g}{N} \approxeq \frac{1}{\sqrt{q} - 1}$ and ignoring constants, we have
	\begin{align*}
	|A| &  \ge q^{ug} \left (1 - \dfrac{1}{q} \right )^{N} -  \left(\sqrt{q} - 1 \right)^{(u - 1)g + 1} q^{g} \\
	& \ge q^{ug} \left (1 - \dfrac{1}{q} \right )^{N} -  \sqrt{q}^{(u - 1)g + 1} q^{g} \\
	& =  q^{ug} \left (1 - \dfrac{1}{q} \right )^{N}  - q^{\frac{(u + 1)g + 1}{2}}.
	\end{align*}
	Notice that if we argue $q^{ug} \left (1 - \frac{1}{q} \right )^{N}  \ge q^{(u - 1)g}$ and $q^{\frac{(u + 1)g + 1}{2}} < q^{(u - 1)g}$, then we can conclude that there are at least $q^{(u - 1)g}$ codewords with full Hamming weight (since the first term would dominate). First, we prove
	\begin{align*}
	q^{ug} \left (1 - \frac{1}{q} \right )^{N}  \ge q^{(u - 1)g}.
	\end{align*}
	To this end, we first rewrite the above inequality as follows:
	\begin{align*}
	\left (1 - \frac{1}{q} \right )^{N} \ge q^{-g}.
	\end{align*}
	Using $1 - \frac{1}{q} \ge \me^{-2/q}$, we only need to argue that
	\begin{align*}
	\me^{-\frac{2N}{q}} \ge q^{-g} = \me^{-g \log{q}}.
	\end{align*} 
	Comparing only the exponents and reversing the inequality by multiplying with $-1$, we have
	\begin{align*}
	\frac{2N}{q} \le g \log{q},
	\end{align*}
	which is always true since $\frac{g}{N} \approxeq \frac{1}{\sqrt{q} - 1}$.
	We now prove
	\begin{align*}
	q^{\frac{(u + 1)g + 1}{2}} < q^{(u - 1)g}.
	\end{align*}
	Comparing the exponents, we have
	\begin{align*}
	(u - 1)g - \frac{(u + 1)g + 1}{2} >  0.
	\end{align*}
	Note that the above inequality is always true since $u > 3$ (by assumption). Thus, with $u > 3$ and $K = ug$, we have $q^{(u - 1)g} = q^{\frac{u - 1}{u} K}$ codewords with full Hamming weight. To complete the proof, we need to satisfy the distance constraint as well i.e. $D  = \delta N \ge \left(1 - \frac{\epsilon}{r - 1}\right)N$. In particular, we have 
	\[D \ge N - (K + g -1) \ge \left(1 - \frac{\epsilon}{r - 1}\right) N.\] 
	We can simplify the above expression as follows:
	\[1 - \frac{K}{N} - \frac{g}{N} + \frac{1}{N} \ge \left(1 - \frac{\epsilon}{r - 1}\right).\]
	Observe that we can ignore $\frac{1}{N}$ since $N$ is reasonably large and substitute $\frac{K}{N} = \frac{ug}{N}$ and $\frac{g}{N} \approxeq \frac{1}{\sqrt{q} - 1}$ to get
	\[\frac{u + 1}{\sqrt{q} - 1} \le \frac{\epsilon}{r - 1}.\]
	Writing the inequality in terms of $q$ we get
	\[q > 2 (u + 1)^{2} \frac{r^{2}}{\epsilon^{2}} \ge \left((u + 1) \frac{r - 1}{\epsilon} + 1\right)^{2}.\]
	Once we pick such a $q$, we have shown that $\mathcal{C}^{II}$ has at least $q^{\frac{u - 1}{u} K}$ codewords having Hamming weight, completing the proof.
\end{proof}
\paragraph{} Finally, we present an explicit construction of $\mathcal{C}$ by first choosing a $\mathcal{C}^{II}$ from Theorem~\ref{theorem:ag1} and a corresponding $\mathcal{C}^{I}$ from Construction~\rom{4} in~\cite{YB17}.
\begin{theorem} \label{Lemma:inffamily}
Given positive integers $r$, $s \le r$, $u > 3$ and an $\epsilon > 0$, there exists a constant square prime power $q: q > 2 (u + 1)^{2} \cdot \frac{r^2}{\epsilon^{2}}$ such that for infinite values of $\mathcal{N}$, there exists an $\left(\mathcal{N}, \mathcal{K} = \mathcal{N} - r, \mathcal{T} = \mathcal{N} - s, \mathcal{L}\right)_{\mathbb{B}}$ $\epsilon$-MSR code satisfying the $\left(\mathcal{T} = \mathcal{N} - s, \mathcal{T}' \le \mathcal{N} - \mathcal{N}^{\frac{u - 1}{u}} \right)$-repair property. Moreover, the sub-packetization $\mathcal{L} = O_{s, q, u}(\log{\mathcal{N}})$ and the required field size $|\mathcal{B}|$ scales as $O_{r, q}(\mathcal{N})$. 
\end{theorem}
\begin{proof}
To this end, we first choose a prime power $q$ such that $q > 2 (u + 1)^{2} \cdot \frac{r^2}{\epsilon^2}$ for fixed values of $r$, $u$ and $\epsilon$. Then, from Theorem~\ref{theorem:ag1}, we can always construct a code $\mathcal{C}^{II} =  (N, M = q^{K}, D = \delta N)_{q}$ having $q^{\frac{u - 1}{u} K}$ codewords of Hamming weight $N$. We combine this with the $(n = q, k = q - r, s \le r, \ell = s^{q})_{\mathbb{B}}$ MSR code with $t = s + k - 1$ from~\cite{YB17} as described above. This gives us an $\epsilon$-MSR code with length $\mathcal{N} = q^{K}$ and sub-packetization level $\mathcal{L} = N \ell = N s^{q}$. Observe that $\mathcal{N} = q^{K} = q^{u g}$ and taking $\log$ on both sides, we have $\log{\mathcal{N}} = u g \log{q}$. We can now substitute $g \approxeq \frac{N}{\sqrt{q} - 1}$ to get $\log{\mathcal{N}} = u \frac{N}{\sqrt{q} - 1} \log{q}$. Now, using $N = \frac{\mathcal{L}}{s^{q}}$, we have
\[\log{\mathcal{N}}  = u \frac{\mathcal{L}}{s^{q} (\sqrt{q} - 1)} \log{q}.\]
Since $q$, $s$ and $u$ are constants, we have that $\mathcal{L} = O_{s, q, u}(\log{\mathcal{N}})$.

From the construction of $\mathcal{C}$, it follows that the scalars $\{\lambda_{i, j}\}_{i \in [q], j \in [s]}$ of $\mathcal{C}^{I}$ can be obtained from a finite field of size at least $qs$ and the scalars $\{\sigma_{i}^{j - 1}\}_{i \in [M], j \in [r]}$ of $\mathcal{C}$ can be obtained from a finite field of size at least $Mr = q^{K}r$. Combining this with the condition $s < r$, we need a field size of at least $q^{K}qr + 1$. Thus, for constant $r$ and $q$, the required field size $O(q^{K}qr)$ scales as $O_{r, q}(\mathcal{N})$, completing the proof.
\end{proof}
\section{Conclusion}
We present an explicit construction of $\epsilon$-MSR codes satisfying the $\mathcal{T}$-optimal repair property, resolving an open question posed in~\cite{RTGE17} ,\cite{RTGE17J} and~\cite{GR17}. The obtained codes ensure load balancing among the contacted code blocks during the repair process. We see two major directions to extend this work: one is to obtain a construction where the repair procedure can contact any subset of $\mathcal{T}$ helper nodes or prove a lower bound showing that some compulsory code blocks are always necessary. The second is to extend this construction to $\epsilon$-MSR codes which can repair multiple erasures, which is part of our ongoing work.

\section*{Acknowledgments}
We thank Vijay Kumar, Sergey Yekhanin, Myna Vajha, Vinayak Ramkumar and Chaoping Xing for helpful discussions. We thank Microsoft Research, Bangalore and School of Physical and Mathematical Sciences, Nanyang Technological University, Singapore for their hospitatilty. VG likes to thank NSF for their generous support through the grant CCF-1563742.
\bibliographystyle{acm}
\bibliography{sai}
\appendix
\section{Missing Details in Section~3}
\subsection{Ye-Barg Construction with $t$-Optimal Repair Property (from~\cite{YB17})} \label{Ye-Barg-d-optimal}
\subsubsection{Preliminaries and Definition}
\begin{definition} [$t$-optimal repair property] \label{Definition:msrtopt}
	For a fixed value of $k \le t < n - 1$, an MSR code $\mathcal{C} = [n, k\ell, d_{min} = n - k + 1, \ell = s^{n}]$ with $s = t - k + 1$ is said to have the {\em $t$-optimal repair property} if and only if the exact repair of any node $i \in [n]$ can be accomplished by receiving $\frac{\ell}{s}$ symbols (over $\mathbb{B}$) from any set of $t$ \emph{helper nodes}.	
\end{definition}
\subsubsection{Construction}
\begin{construction} \label{construction:ye-barg-d}
	Let $F$ be a finite field of size $|F| \ge sn$, and $\ell = s^{n}$. Let $\{\lambda_{i, j}\}_{i \in [n], j \in [s]}$ be $sn$ distinct elements in $F$. Let $\mathcal{C}^{I} \in F^{\ell n}$ be an $(n, k, t, \ell)$ array code with nodes $C^{I}_i \in F^{\ell}$, $i  \in [n]$, where each $C^{I}_i$ is a column vector indexed by $\{c^{I}_{i, b}: b \in[0, \ell - 1]\}$. $\mathcal{C}^{I}$ is defined in the following parity-check form:
	\begin{equation} \label{eq:ybs-pc1}
	\mathcal{C}^{I} = \{(C^{I}_i)_{i \in [n]}: \sum_{i = 1}^{i = n} H_{j, i} \cdot C^{I}_i = 0, \text{\space} j  \in [r]\}, 
	\end{equation}
	where $H_{j, i}: j \in [r], i \in [n]$ are $\ell \times \ell$ matrices over $F$. Given positive integers $r$ and $n$, the $(n, k = n - r, k \le t < n - 1, \ell = s^{n})$ array code with $s = t - k + 1$, the code $\mathcal{C}$ is defined by
	\begin{equation} 
	H_{j, i} = H_{i}^{j - 1}, j \in [r], i \in [n],
	\end{equation}
	where $H_{i} = \sum_{b = 0}^{\ell - 1}\lambda_{i, b_i} \cdot e_b \cdot e_{b}^{T}$ (with a slight abuse of notation) for every $i \in [n]$. Here, $\{e_b : b \in [0, \ell - 1]\}$ is the standard basis of $F^{\ell}$ over $F$, and $b_j$ is the $j$-th digit from the right in the representation of $b$ in the $s$-ary form, $b = (b_n, b_{n  -1}, \dots, b_1)$.
\end{construction}
\begin{lemma} \label{Lemma:Ye-Barg-Simple-d-Repair}
	The code $\mathcal{C}$ given by Construction~\ref{construction:ye-barg-d} has the $t$-optimal repair property (Definition~\ref{Definition:msrtopt}).
\end{lemma}
\begin{proof}
	Without loss of generality, let's start by assuming that $C^{I}_{1}$ fails. We can write the $r\ell$ parity-check equations corresponding to $i \in [n]$ as follows:
	\begin{align*}
		H_{1, j} C^{I}_{1} + \sum_{i = 2}^{i = n} H_{i, j} C^{I}_{i} = 0 \text{ for every } j \in [r].
	\end{align*}
	We first recover $s$ symbols of the form $(c^{I}_{1, b(1, k)})_{k \in [0, s - 1]}$ for a fixed value of $b$ in the range $[0, \ell - 1]$. Note that $b$ is a $s$-ary representation of length $n$ and $b(i, u) = (b_{n}, \dots, b_{i+1}, u, b_{i - 1}, \dots, b_{1})$. Consider the $rs$ parity check equations corresponding to these $s$ symbols, given by
	\begin{align*}
		\lambda_{1, k}^{j - 1} c^{I}_{1, b(1, k)} + \sum_{i = 2}^{i = n} \lambda_{i, b_{i}}^{j - 1} c^{I}_{i, b(1, k)} = 0 \text{ for every } j \in [r], k \in [0, s - 1].
	\end{align*}
	We can sum this equation over all values of $k \in [0, s - 1]$ to get
	\begin{align*}
		\sum_{k = 0}^{s - 1} \lambda_{1, k}^{j - 1} c^{I}_{1, b(1, k)} + \sum_{k = 0}^{s - 1} \sum_{i = 2}^{i = n} \lambda_{i, b_{i}}^{j - 1} c^{I}_{i, b(1, k)} = 0 \text{ for every } j \in [r].
	\end{align*}
	Propagating the sum inside, we have
	\begin{align*}
		\sum_{k = 0}^{s - 1} \lambda_{1, k}^{j - 1} c^{I}_{1, b(1, k)} + \sum_{i = 2}^{i = n} \lambda_{i, b_{i}}^{j - 1} \sum_{k = 0}^{s - 1} c^{I}_{i, b(1, k)} = 0 \text{ for every } j \in [r].
	\end{align*}
	For every $i \in [2, n]$, let $\mu_{i, 1}^{(b)} = \sum_{k = 0}^{s - 1} c^{I}_{i, b(1, k)}$. We can substitute this in the above equation to obtain
	\begin{align*}
		\sum_{k = 0}^{s - 1} \lambda_{1, k}^{j - 1} c^{I}_{1, b(1, k)} + \sum_{i = 2}^{i = n} \lambda_{i, b_{i}}^{j - 1} \mu_{i, 1}^{(b)} = 0 \text{ for every } j \in [r].
	\end{align*}
	We can write the above $r$ equations in the following matrix form:
	\begin{equation} \label{eq:ybs-pc6}
	L_1 + L_3 = 0,
	\end{equation}
	where $L_1 = E_{L_1} F_{L_1}, L_3 = E_{V} F_{V}$. In particular,
	\begin{align*}
		E_{L_{1}} = (\lambda_{1, k}^{j - 1})_{j \in [r], k \in [0, s - 1]}, F_{L_1} = (c^{1}_{1, b(1, k)})_{k \in [0, s - 1]},
		E_{V} = (\lambda_{i, b_{i}}^{j - 1})_{j \in [r], i \in [2, n]}, F_{V} = (\mu^{(b)}_{i, 1})_{i \in [2, n]}.
	\end{align*}
	We use a polynomial interpolation argument on~\eqref{eq:ybs-pc6}, whose overview we provide here. Recall that our goal is to compte the number of symbols we need to download from the set $F_{V}$ in order to recover $F_{L_1}$. To do this, we first make $E_{L_1}$ zero and show that we can recover $F_{V}$ from any of its $|V| - (r - s)$ symbols. Then, we recover $F_{L_1}$ using $F_{V}$. Observe that we can recover $s$ symbols by contacting any $t$ remaining nodes and downloading $1$ scalar from each of them. 
	
	Finally, we have that in order to recover $C^{I}_{1}$ completely, we need to download a total of $t \frac{\ell}{s}$ symbols (over $F$) by contacting any $t$ nodes respectively, proving that $\mathcal{C}$ satisfies the $t$-optimal repair property.
\end{proof}
\subsection{Proof of MDS Property} \label{LemmaRef:d11}
\begin{lemma} \label{Lemma:d11}
	The code $\mathcal{C}$ is an MDS code.
\end{lemma}
\begin{proof}
	To argue that $\mathcal{C}$ is a MDS code, we show that any $rN\ell \times rN\ell$ sub-matrix of $\mathcal{H}$ consisting of $r = n - k$ thick columns of $\mathcal{H}$ corresponding to any $r$ distinct code blocks is full rank. Observe that by construction, the $r$ code blocks are indexed by $r$ different codewords in $\mathcal{C}^{II}$. Let the index of these $r$ codewords be denoted by: 
	\[E = \{e_{1}, \dots, e_{r}\}.\]
	The $rN\ell \times rN\ell$ sub-matrix of $\mathcal{H}$ that corresponds to the code blocks indexed by these codewords takes the following form:
	\begin{align*}
	\mathcal{H}_{E} & = \begin{bmatrix}
	\mathcal{H}_{e_{1}} & \mathcal{H}_{e_{2}} & \dots & \mathcal{H}_{e_{r}}
	\end{bmatrix}.
	\end{align*}
	Using~\eqref{eq:main_pc2}, the above equation can be written as
	\begin{align*}
	H_{E} & = \begin{bmatrix}
	Diag(H_{1, a_{e_1, 1}}, \dots, H_{1, a_{e_1, N}}) & \dots & Diag(H_{1, a_{e_r, 1}}, \dots, H_{1, a_{e_r, N}})\\
	\sigma_{e_1} \cdot Diag(H_{2, a_{e_1, 1}}, \dots, H_{2, a_{e_1, N}}) & \dots & \sigma_{e_r} \cdot Diag(H_{2, a_{e_r, 1}}, \dots, H_{2, a_{e_r, N}})\\
	\vdots & \vdots & \vdots \\
	\sigma_{e_1}^{r - 1} \cdot Diag(H_{r, a_{e_1, 1}}, \dots, H_{r, a_{e_1, N}}) & \dots & \sigma_{ e_r}^{r - 1} \cdot Diag(H_{r, a_{e_r, 1}}, \dots, H_{r, a_{e_r, N}}) \\
	\end{bmatrix}.
	\end{align*}
	Given the block diagonal structure of the $N\ell \times N\ell$ sub-matrices in $H_{E}$ (there are $r^{2}$ in total), we only need to argue that the following matrix is full rank:
	\begin{equation} 
	\mathbf{U}_{E, i} = \begin{bmatrix}
	H_{1, a_{e_1, i}} & H_{1, a_{e_2, i}} & \dots & H_{1, a_{e_{r}, i}}  \\
	\sigma_{e_1} \cdot  H_{2, a_{e_1, i}} & \sigma_{e_2} \cdot H_{2, a_{e_2, i}} & \dots & \sigma_{e_r} \cdot H_{2, a_{e_r, i}} \\
	\vdots & \vdots & \vdots & \vdots \\
	\sigma_{e_1}^{r - 1} \cdot  H_{r, a_{e_1, i}} & \sigma_{e_2}^{r - 1} \cdot H_{r, a_{e_2, i}} & \dots & \sigma_{e_r}^{r - 1} \cdot H_{r, a_{e_r, i}}
	\end{bmatrix}
	\end{equation}
	for every $i \in [N]$. Recalling the definition of $H_{i, j}$ for all $i \in [r], j \in [r]$ (from Appendix~\ref{Ye-Barg-d-optimal}), we see that they are diagonal blocks. Similar to the proof of Theorem \rom{4}.2 in~\cite{YB17}, we can rearrange the rows and columns of the matrix $\mathbf{U}_{E, i}$ to obtain a block diagonal matrix, where the diagonal blocks are Vandermonde matrices. Therefore, the matrix $U_{R, i}$ is a full rank matrix for every $i \in [N]$, implying that $\mathcal{C}$ is a MDS code, completing the proof.
\end{proof}
\subsection{Repair Bandwidth: Polynomial Interpolation Argument} \label{LemmaRef:PI1}
\paragraph{} To this end, we first construct a matrix $P$ of polynomial coefficients and left-multiply it with~\eqref{eq:ourd-pc12} to get
\[ P L_1 + P L_2  + P L_3 = 0.\]
Then, we argue that $P L_1 = 0$, $P L_2$ is a constant vector and $P L_3$ is a matrix with rank $r - s$ (i.e. full rank). In particular, this would imply that by downloading any subset of size $|V| - r - s$ from $F_{V}$, the remaining symbols can be recovered using the equation $P L_3 = - P L_2$. After this recovery, we can solve the equation $L_1 = - L_2 - L_3$ to recover $F_{L_1}$, which is our final goal.

\paragraph{} To construct $P$, we first define polynomials $p_{0}(x) = \prod_{u = 0}^{s - 1} (x - \sigma_{1} \cdot \lambda_{a_{1, 1}, u})$ and $p_{i}(x) = x^{i} p_{0}(x)$ for $i = 0, 1, \dots, r - s - 1$.  Notice that $t < n - 1$ and as a result, we have $r - s - 1 \ge 0$. Since the degree of $p_{i}(x)$ is less than $r$ for all $i =0, 1, \dots, r - s - 1$, we can write
\[p_{i}(x) = \sum_{j = 0}^{r - 1} p_{i,j} x^{j}.\]
We now define the $(r - s) \times r$ matrix
\begin{align*}
P = \begin{bmatrix}
p_{0, 0} & p_{0, 1} & \dots & p_{0, r - 1} \\
p_{1, 0} & p_{1, 1} & \dots & p_{1, r -1} \\
\vdots & \vdots & \vdots & \vdots \\
p_{r - s - 1, 0} & p_{r - s - 1, 1} & \dots & p_{r - s - 1, r - 1}
\end{bmatrix}.
\end{align*}
We multiply $P$ with the left-hand side of $L_1$ to get
\begin{align*}
P \cdot L_{1} & = \begin{bmatrix}
p_{0, 0} & \dots & p_{0, r - 1} \\
p_{1, 0} & \dots & p_{1, r -1} \\
\vdots & \vdots & \vdots \\
p_{r - s - 1, 0} & \dots & p_{r - s - 1, r - 1}
\end{bmatrix}
\begin{bmatrix}
1 & \dots & 1 \\
\sigma_{1} \cdot \lambda_{a_{1, 1}, 0} & \dots & \sigma_{1} \cdot \lambda_{a_{1, 1}, s - 1} \\
\sigma_{1}^{2} \cdot \lambda_{a_{1, 1}, 0}^2 & \dots & \sigma_{1}^{2} \cdot \lambda_{a_{1, 1}, s - 1}^2 \\
\vdots & \vdots & \vdots \\
\sigma_{1}^{r - 1} \cdot \lambda_{a_{1, 1}, 0}^{r -1} & \dots & \sigma_{1}^{r - 1} \cdot \lambda_{a_{1, 1},  s - 1}^{r - 1}
\end{bmatrix}
\begin{bmatrix}
c^{1}_{1, b(1, 0)} \\
c^{1}_{1, b(1, 1)}  \\
\vdots \\
c^{1}_{1, b(1, s - 1)}
\end{bmatrix}.
\end{align*}
Observe that the above equation simplifies to
\begin{align*}
P \cdot L_1 & = \begin{bmatrix}
p_{0}(\sigma_{1} \cdot \lambda_{a_{1, 1}, 0}) & \dots & p_{0}(\sigma_{1} \cdot \lambda_{a_{1, 1}, s - 1}) \\
p_{1}(\sigma_{1} \cdot \lambda_{a_{1, 1}, 0}) & \dots & p_{1}(\sigma_{1} \cdot \lambda_{a_{1, 1}, s - 1}) \\
\vdots & \vdots & \vdots \\
p_{r - s - 1}(\sigma_{1} \cdot \lambda_{a_{1, 1}, 0}) & \dots & p_{r - s - 1}(\sigma_{1} \cdot \lambda_{a_{1, 1}, s - 1})
\end{bmatrix} 
\begin{bmatrix}
c^{1}_{1, b(1, 0)} \\
c^{1}_{1, b(1, 1)}  \\
\vdots \\
c^{1}_{1, b(1, s - 1)}
\end{bmatrix}
= 0,
\end{align*}
where the final equality follows from our definition of $p_i : i \in [0, r - s- 1]$. We still need to multiply $P$ with $L_2$ and $L_3$ respectively, which we do one-by-one.
\begin{align*}
P \cdot L_{2} & = \begin{bmatrix}
p_{0, 0} & \dots & p_{0, r - 1} \\
p_{1, 0} & \dots & p_{1, r -1} \\
\vdots & \vdots & \vdots \\
p_{r - s - 1, 0} & \dots & p_{r - s - 1, r - 1}
\end{bmatrix}
\mathlarger{
	\sum_{q_i \in Q}}
\begin{bmatrix}
1 & \dots & 1 \\
\sigma_{q_i} \cdot \lambda_{a_{1, 1}, 0} & \dots & \sigma_{q_i} \cdot \lambda_{a_{1, 1}, s - 1} \\
\vdots & \vdots & \vdots \\
\sigma_{q_i}^{r - 1} \cdot \lambda_{a_{1, 1}, 0}^{r -1} & \dots & \sigma_{q_i}^{r - 1} \cdot \lambda_{a_{1, 1}, s - 1}^{r - 1}
\end{bmatrix} 
\begin{bmatrix}
c^{q_i}_{1, b(1, 0)} \\
c^{q_i}_{1, b(1, 1)}  \\
\vdots \\
c^{q_i}_{1, b(1, s - 1)}
\end{bmatrix}
\\
& = \mathlarger{\sum_{q_i \in Q}}
\begin{bmatrix}
p_{0}(\sigma_{q_i} \cdot \lambda_{a_{1, 1}, 0})  & \dots & p_{0}(\sigma_{q_i} \cdot \lambda_{a_{1, 1}, s - 1}) \\
p_{1}(\sigma_{q_i} \cdot \lambda_{a_{1, 1}, 0})  & \dots & p_{1}(\sigma_{q_i} \cdot \lambda_{a_{1, 1}, s - 1}) \\
\vdots & \vdots & \vdots \\
p_{r - s - 1}(\sigma_{q_i} \cdot \lambda_{a_{1, 1}, 0}) & \dots & p_{r - s - 1}(\sigma_{q_i} \cdot \lambda_{a_{1, 1}, s - 1})
\end{bmatrix} 
\begin{bmatrix}
c^{q_i}_{1, b(1, 0)} \\
c^{q_i}_{1, b(1, 1)} \\
\vdots \\
c^{q_i}_{1, b(1, s - 1)}
\end{bmatrix}
\\
& = \mathlarger{\sum_{q_i \in Q}}
\begin{bmatrix}
p_{0}(\sigma_{q_i} \cdot \lambda_{a_{1, 1}, 0}) & \dots & p_{0}(\sigma_{q_i} \cdot \lambda_{a_{1, 1}, s - 1}) \\
p_{0}(\sigma_{q_i} \cdot \lambda_{a_{1, 1}, 0}) \sigma_{q_i} \lambda_{a_{1, 1}, 0} & \dots & p_{0}(\sigma_{q_i} \cdot \lambda_{a_{1, 1}, s - 1}) \sigma_{q_i} \lambda_{a_{1, 1}, s - 1}\\
\vdots & \vdots & \vdots \\
p_{0}(\sigma_{q_i} \cdot \lambda_{a_{1, 1}, 0}) (\sigma_{q_i} \lambda_{a_{1, 1}, 0})^{r - s - 1}& \dots & p_{0}(\sigma_{q_i} \cdot \lambda_{a_{1, 1}, s - 1}) (\sigma_{q_i} \lambda_{a_{1, 1}, s - 1})^{r - s - 1}
\end{bmatrix}
\begin{bmatrix}
c^{q_i}_{1, b(1, 0)} \\
c^{q_i}_{1, b(1, 1)} \\
\vdots \\
c^{q_i}_{1, b(1, s - 1)}
\end{bmatrix}. 
\end{align*}
Note that for all $q_i \in Q$, we need to download all the $s$ symbols in $c^{q_i}_{1}$. Thus, $P \cdot L_2$ results in a $(r - s) \times 1$ constant vector, which we call $\mathbf{y}$.  
\begin{align*}
P \cdot L_3 &= \begin{bmatrix}
p_{0, 0} & \dots & p_{0, r - 1} \\
p_{1, 0} & \dots & p_{1, r -1} \\
\vdots & \vdots & \vdots \\
p_{r - s - 1, 0} & \dots & p_{r - s - 1, r - 1}
\end{bmatrix}
\begin{bmatrix}
1 & \dots & 1 \\
\sigma_{v_1} \cdot \lambda_{r_1, b_{r_1}} & \dots & \sigma_{v_{|V|}} \cdot \lambda_{r_{|R|}, a_{b_{|R|}}} \\
\vdots & \vdots & \vdots \\
\sigma_{v_1}^{r - 1} \cdot \lambda_{r_1, b_{r_1}}^{r - 1} & \dots & \sigma_{v_{|V|}}^{r - 1} \cdot \lambda_{r_{|R|}, b_{r_{|R|}}}^{r - 1}
\end{bmatrix} 
\begin{bmatrix}
\mu^{(b)}_{v_1, 1, 1} \\
\mu^{(b)}_{v_2, 1, 1} \\
\vdots \\
\mu^{(b)}_{v_{|V|}, 1, 1}
\end{bmatrix}
\\
& = \begin{bmatrix}
p_{0}(\sigma_{v_1} \cdot \lambda_{r_1, b_{r_1}}) & \dots & p_{0}(\sigma_{v_{|V|}} \cdot \lambda_{r_{|R|}, b_{r_{|R|}}}) \\
p_{1}(\sigma_{v_1} \cdot \lambda_{r_1, b_{r_1}}) & \dots & p_{1}(\sigma_{v_{|V|}} \cdot \lambda_{r_{|R|}, b_{r_{|R|}}}) \\
\vdots & \vdots & \vdots \\
p_{r - s - 1}(\sigma_{v_1} \cdot \lambda_{r_1, b_{r_1}}) & \dots & p_{r -s - 1}(\sigma_{v_{|V|}} \cdot \lambda_{r_{|R|}, b_{r_{|R|}}})
\end{bmatrix} 
\begin{bmatrix}
\mu^{(b)}_{v_1, 1, 1} \\
\mu^{(b)}_{v_2, 1, 1} \\
\vdots \\
\mu^{(b)}_{v_{|V|}, 1, 1}
\end{bmatrix}
\\
& = \begin{bmatrix}
p_{0}(\sigma_{v_1} \cdot \lambda_{r_1, b_{r_1}}) & \dots & p_{0}(\sigma_{v_{|V|}} \cdot \lambda_{r_{|R|}, b_{r_{|R|}}}) \\
p_{0}(\sigma_{v_1} \cdot \lambda_{r_1, b_{r_1}}) \sigma_{v_1} \lambda_{r_1, b_{r_1}} & \dots & p_{0}(\sigma_{v_{|V|}} \cdot \lambda_{r_{|R|}, b_{r_{|R|}}}) \sigma_{v_{|V|}} \lambda_{r_{|R|}, b_{r_{|R|}}}\\
\vdots & \vdots & \vdots \\
p_{0}(\sigma_{v_1} \cdot \lambda_{r_1, b_{r_1}}) (\sigma_{v_1}\lambda_{r_1, b_{r_1}})^{r - s -1} & \dots & p_{0}(\sigma_{v_{|V|}} \cdot \lambda_{r_{|R|}, b_{r_{|R|}}})(\sigma_{v_{|V|}}\lambda_{r_{|R|}, b_{r_{|R|}}})^{r - s - 1}
\end{bmatrix}
\begin{bmatrix}
\mu^{(b)}_{v_1, 1, 1} \\
\mu^{(b)}_{v_2, 1, 1} \\
\vdots \\
\mu^{(b)}_{v_{|V|}, 1, 1}
\end{bmatrix}.
\end{align*}
Observe that $p_{0}(\sigma_{v_1} \cdot \lambda_{r_1, b_{r_1}}), p_{0}(\sigma_{v_2} \cdot \lambda_{r_2, b_{r_2}}), \dots, p_{0}(\sigma_{v_{|V|}} \cdot \lambda_{r_{|R|},b_{r_{|R|}}})$ are all nonzero and $P \cdot L_3$ is a full-rank matrix. Further, all $(r - s) \times (r - s)$ sub-matrices of $P \cdot L_3$ also have full rank, which follows from its Vandermonde-like structure. Finally, we have
\begin{align*}
\begin{bmatrix}
p_{0}(\sigma_{v_1} \cdot \lambda_{r_1, b_{r_1}}) & \dots & p_{0}(\sigma_{v_{|V|}} \cdot \lambda_{r_{|R|}, b_{r_{|R|}}}) \\
p_{0}(\sigma_{v_1} \cdot \lambda_{r_1, b_{r_1}}) \sigma_{v_1} \lambda_{r_1, b_{r_1}} & \dots & p_{0}(\sigma_{v_{|V|}} \cdot \lambda_{r_{|R|}, b_{r_{|R|}}}) \sigma_{v_{|V|}} \lambda_{r_{|R|}, b_{r_{|R|}}}\\
\vdots & \vdots & \vdots \\
p_{0}(\sigma_{v_1} \cdot \lambda_{r_1, b_{r_1}}) (\sigma_{v_1}\lambda_{r_1, b_{r_1}})^{r - s -1} & \dots & p_{0}(\sigma_{v_{|V|}} \cdot \lambda_{r_{|R|}, b_{r_{|R|}}})(\sigma_{v_{|V|}}\lambda_{r_{|R|}, b_{r_{|R|}}})^{r - s - 1}
\end{bmatrix}
\begin{bmatrix}
\mu_{v_1, 1, 1}^{(b)} \\
\mu_{v_2, 1, 1}^{(b)} \\
\vdots \\
\mu_{v_{|V|}, 1, 1}^{(b)}
\end{bmatrix}
= - \mathbf{y}.
\end{align*}
Observe that by downloading any subset of size $|V| - r - s$ from $F_{V} = (\mu_{v_1, 1, 1}^{(b)}, \dots, \mu_{v_{|V|}, 1, 1}^{(b)})$, the remaining symbols can recovered from the above equation. As stated earlier, we can now solve the equation $L_1 = - L_2 - L_3$ to recover all the $s$ symbols in $F_{L_1}$, completing the argument.

\subsection{AG codes} \label{agcodes}
Let $\mathcal{X}$ be a smooth, projective, absolutely irreducible curve of genus $g$ defined over $\mathbb{F}_{q}$.  We denote by $\mathbb{F}_{q}(\mathcal{X})$ the function field of $\mathcal{X}$. An element of $\mathbb{F}_{q}(\mathcal{X})$ is called a function. The normalized discrete valuation corresponding to a point $P$ of $\mathbb{F}_{q}(\mathcal{X})$ is written as $\varv$. A point $P$ is said to be $\mathbb{F}_{q}$ if $P^{\sigma} = P$ for all $\sigma$ in the Galois group $\gal(\bar{\mathbb{F}}_{q}/\mathbb{F}_{q})$. Likewise, a divisor $G = \sum_{P} m_{P} P$ is said to be $\mathbb{F}_{q}$-rational if $G^{\sigma} = \sum_{P} m_{P} P^{\sigma} = G$ for all $\sigma$ in the Galois group $\gal(\bar{\mathbb{F}}_{q}/\mathbb{F}_{q})$. 

For an $\mathbb{F}_{q}$-rational divisor $G$, the Riemann-Roch space associated to $G$ is
\[\mathcal{L}_{\mathbb{F}_{q}} (G) = \{ f \in \mathbb{F}_{q}(\mathcal{X}) \setminus \{0\} : \Div{(f)} +G \ge 0 \} \cup \{0\}.\]
Then $\mathcal{L}_{\mathbb{F}_{q}}$ is a finite-dimensional vector space over $\mathbb{F}_{q}$ and we denote its dimension by $\ell(G)$. By the Riemann-Roch theorem, we have
\[\ell(G) \ge \deg(G) + 1 - g,\]
where the equality holds if $\deg(G) \ge 2g - 1$.

Let $P_1, \dots, P_{N}$ be pairwise distinct $\mathbb{F}_{q}$-rational points of $\mathcal{X}$ and $R = P_1 + \dots + P_{N}$. Choose an $\mathbb{F}_{q}$-rational divisor $G$ in $\mathcal{X}$ such that $\supp(G) \cap \supp(R) = \emptyset$, and a vector $\mathbf{v} = (v_1, \dots, v_{N})$ such that $v_i \in (\mathbb{F}_{q})^{*}$ for every $i \in [N]$. Then, $\varv_{P_i}(f) \ge 0$ for all $i \in [N]$ and any $f \in \mathcal{L}_{\mathbb{F}_{q}}(G)$. 

Consider the map
\[\psi : \mathcal{L}(G) \mapsto \mathbb{F}_{q}^{n}, f \mapsto (v_1 f(P_1), \dots, v_{N}f(P_N)).\]
Obviously the image of $\psi$ is a subspace of $\mathbb{F}_{q}^{N}$. The image of $\psi$ is denoted as $\mathcal{C}^{II} = C_{\mathcal{L}}(R, G, \mathbf{v})$ which is called an algebraic-geometry code (AG code hereon). If $\deg(G) < N$, then $\psi$ is an embedding and we have $\dim(C_{\mathcal{L}}(R, G, \mathbf{v})) = \ell(G)$. By the Riemann-Roch theorem, we can estimate the parameters of an AG code. 
\begin{proposition} \label{prop:r-r}
	$C_{\mathcal{L}}(R, G, \mathbf{v})$ is an $[N, K, D]$-linear code $\mathbb{F}_{q}$ with parameters
	\[K = \ell(G) - \ell(G- R), D \ge N - \deg(G).\]
	If $G$ satisfies $g \le \deg(G) < N$,then
	\[K = \ell(G) \ge \deg(G) - g + 1, D \ge N - \deg(G).\]
	If addtionally $2g - 2 < \deg(G) < N$, then
	\begin{equation} \label{eq:agdimension}
	K = \deg(G) - g + 1.
	\end{equation}
\end{proposition}

\end{document}